\documentclass[conference,letterpaper,onecolumn]{IEEEtran}

\usepackage[margin=1.125in]{geometry}
\usepackage{graphicx} 
\usepackage[cmex10]{amsmath}
\usepackage{amssymb, amsthm, graphicx, xcolor, mathtools,cite}
\usepackage{hyperref}

\theoremstyle{definition}
\newtheorem{definition}{Definition}
\newtheorem{theorem}{Theorem}
\newtheorem{lemma}{Lemma}

\newtheorem{claim}{Claim}
\newtheorem{remark}{Remark}
\newtheorem{observation}{Observation}
\definecolor{ferngreen}{rgb}{0.11, 0.44, 0.26}
\definecolor{iris}{rgb}{0.35, 0.31, 0.81}
\definecolor{BurntOrange}{HTML}{F7921D}
\definecolor{orange(colorwheel)}{rgb}{1.0, 0.5, 0.0}
\definecolor{darkpastelblue}{rgb}{0.47, 0.62, 0.8}
\definecolor{yellowgreen}{HTML}{98CC70}
\definecolor{emerald}{rgb}{0.31, 0.78, 0.47}
\definecolor{cadet}{rgb}{0.33, 0.41, 0.47}
\definecolor{notcharcoal}{rgb}{0.1, 0.22, 0.38}

\usepackage{bbm}
\newcommand{\sett}[1]{ \{ #1 \} }

\newcommand{\enc}[1]{\mathrm{Enc}_{#1}}
\newcommand{\dec}[1]{\mathrm{Dec}_{#1}}

\newcommand{\maj}{\mathrm{Maj}}

\newcommand{\calg}{\mathcal{G}}
\newcommand{\calc}{\mathcal{C}}
\newcommand{\calw}{\mathcal{W}}
\newcommand{\calr}{\mathcal{R}}
\newcommand{\cals}{\mathcal{S}}
\newcommand{\cali}{\mathcal{I}}
\newcommand{\caly}{\mathcal{Y}}
\newcommand{\calu}{\mathcal{U}}
\newcommand{\vhat}{\hat{v}}
\newcommand{\chat}{\hat{c}}
\newcommand{\what}{\hat{w}}
\newcommand{\jhat}{\hat{j}}
\newcommand{\compr}{\texttt{compute-r}}
\newcommand{\ghat}{\hat{g}}
\newcommand{\chunk}{full chunk }  
\newcommand{\chunkn}{full chunk}
\newcommand{\pfail}{P_{\text{fail}}}
\definecolor{purplyblue}{rgb}{0.35, 0.28, 0.98}
\newcommand{\dorsa}[1]{{ \color{purplyblue} [#1 -Dorsa]} }

\newcommand{\calrm}{\mathcal{R}\mathcal{M}}
\newcommand{\ber}{\mathrm{Ber}}
\newcommand{\F}{\mathbb{F}}
\usepackage{algorithm}
\usepackage{algpseudocode}
\newcommand{\gComment}[1]{\Comment{\textcolor{black!30!blue}{#1}}}

\usepackage[utf8]{inputenc} 
\usepackage[T1]{fontenc}
\usepackage{url}
\usepackage{ifthen}
\usepackage{cite}
\usepackage[cmex10]{amsmath} 

\begin{document}
\title{Improved Construction of Robust Gray Codes}

\author{%
  \IEEEauthorblockN{Dorsa Fathollahi}
  \IEEEauthorblockA{Department of Electrical Engineering\\
                    Stanford University\\
                    Stanford, CA, USA\\
                    Email: dorsafth@stanford.edu}
  \and
  \IEEEauthorblockN{Mary Wootters}
  \IEEEauthorblockA{Departments of 
  Computer Science and Electrical Engineering\\
                    Stanford University\\
                    Stanford, CA, USA\\
                    Email: marykw@stanford.edu}
}

\maketitle

\begin{abstract}
A \emph{robust Gray code}, formally introduced by (Lolck and Pagh, SODA 2024), is a Gray code that additionally has the property that, given a noisy version of the encoding of an integer $j$, it is possible to reconstruct $\hat{j}$ so that $|j - \hat{j}|$ is small with high probability.  That work presented a transformation that transforms a binary code $\calc$ of rate $R$ to a robust Gray code with rate $\Omega(R)$, where the constant in the $\Omega(\cdot)$ can be at most $1/4$.  We improve upon their construction by presenting a transformation from a (linear) binary code $\calc$ to a robust Gray code with similar robustness guarantees, but with rate that can approach $R/2$.
\end{abstract}

\section{Introduction}
In \cite{LP24}, Lolck and Pagh introduce the notion of a \textbf{robust Gray code}.  Informally, a robust Gray code $\calg \subseteq \{0,1\}^d$ has an encoding map $\enc{\calg}:\{0,\ldots, N-1\} \to \{0,1\}^d$ that maps integers to bitstrings, with the following desiderata.
\begin{itemize}
    \item \textbf{$\calg$ should be a Gray code.}\footnote{The paper \cite{LP24} also gives a more general definition, where the code should have low \emph{sensitivity}, meaning that $|\enc{\calg}(j) - \enc{\calg}(j+1)|$ is small; however, both their code and our code is a Gray code, so we specialize to that case (in which the sensitivity is $1$).} That is, for any $j \in \{0,\ldots, N-2\}$, $|\enc{\calg}(j) - \enc{\calg}(j+1)| = 1$.
    \item \textbf{$\calg$ should be ``noise robust.''}  Informally, this means that we should be able to approximately recover an integer $j \in \{0, \ldots, N-1\}$ given a noisy version of $\enc{\calg}(j)$.  Slightly more formally, $\calg$ should have a decoding map $\dec{\calg}:\{0,1\}^d \to \{0, \ldots, N-1\}$, so that when $\eta \sim \ber(p)^n$, the estimate $\hat{j} = \dec{\calg}(\enc{\calg}(j) \oplus \eta)$ should be close to $j$ with high probability.
      \item \textbf{$\calg$ should have high rate.}  The rate $\frac{d}{\log N}$ of $\calg$ should be as close to $1$ as possible.
      \item \textbf{$\calg$ should have efficient algorithms.}  Both $\enc{\calg}$ and $\dec{\calg}$ should have running time polynomial (ideally, near-linear) in $d$.
\end{itemize}
Robust Gray codes have applications in differential privacy; see~\cite{LP24,ALP21,ALS22,ACLST21} for more details on the connection. It is worth mentioning that there exist non-binary codes based on the Chinese Remainder Theorem~\cite{XXW20,WX10} that have nontrivial sensitivity, but in our work, we focus on binary codes.

\vspace{.3cm}
\noindent
\textbf{Our Contributions.}  In this paper, we improve upon the construction of \cite{LP24} by giving a construction of a robust Gray code with the same robustness guarantees, but better rate.

More precisely, for $p \in (0,1/2)$, \cite{LP24} give a general recipe for turning a binary error-correcting code $\calc$ with rate $R$ into a robust Gray code $\calg$ with rate $\Omega(R)$, and with the following robustness guarantee:
\begin{equation}\label{eq:robust}
\Pr[ |j - \dec{\calg}(\enc{\calg}(j) + \eta)| \geq t \leq \exp(-\Omega(t)) + \exp(-\Omega(d)) + O(\pfail(\calc)),
\end{equation}
where the probability is over the noise vector $\eta \sim \ber(p)^d$, and $\pfail(\calc)$ is the failure probability of the code $\calc$ on the binary symmetric channel with parameter $p$.

Our main result is a similar transformation that turns a (linear) binary code $\calc$ with good performance on the binary symmetric channel into a robust Gray code $\calg$.  We obtain a similar robustness guarantee as \eqref{eq:robust} (see Theorem~\ref{thm:main} for the precise statement), but with better rate.  Concretely, if the original code $\calc$ has rate $R \in (0,1)$, the rate of the robust Gray code from \cite{LP24} is proven to be $\Omega(R)$, where the constant inside the $\Omega$ approaches $1/4$ when $\calc$ has sublinear distance; this comes from the fact that the a codeword in their final construction involves \emph{four} codewords from $\calc$.  In contrast, under the same conditions, our robust Gray code $\calg$ has rate approaching $R/2$; this is because our construction involves only \emph{two} codewords from $\calc$.  (See Observation~\ref{obs:rate} for the formal statement).  Moreover, if the encoding and decoding algorithms for $\calc$ are efficient, then so are the encoding and decoding algorithms for our construction $\calg$; concretely, the overhead on top of the encoding and decoding algorithms for $\calc$ is $O(d)$ (see Lemma~\ref{lem:running time} for the formal statement). 

As a result, when instantiated with, say, a constant-rate Reed-Muller code or a polar code (both of which have sublinear distance and good performance on the BSC($p$)~(see, e.g., \cite{RP23,A08,GX14, MHR16,WLVG23})), our construction gives efficient robust Gray codes with a rate about two times larger than than previous work, approaching $1/2$.

\vspace{.3cm}
\noindent
\textbf{Main Idea.} The idea of our transformation is quite simple, and follows the same high-level structure as \cite{LP24}.  We begin with our base code $\calc$, and use it to construct an intermediate code $\calw$ (with an appropriate ordering).  Then we add new codewords to $\calw$ to complete it to a Gray code.  For example, if $w_i, w_{i+1}$ are two consecutive codewords in $\calw$, then we will insert $\Delta(w_i, w_{i+1}) - 1$ codewords in between them, iteratively flipping bits to move from $w_i$ to $w_{i+1}$.

The main difference between our construction and that of previous work is how we build and order $\calw$.  First, we use a standard Gray code to construct an \emph{ordering} of the codewords in $\calc$.  Then, we build $\calw$ as follows.  Let $c_i$ be the $i$'th codeword in $\calc$.  Then the $i$'th codeword in $\calw$ is given by
\[ w_i = s_i \circ c_i \circ s_i \circ c_i \circ s_i,\]
where $s_i$ is a short string that is all zeros if $i$ is even and all ones otherwise, and $\circ$ denotes concatenation.  Then we form $\calg$ by interpolating as described above.

Our decoding algorithm ends up being rather complicated, but the idea is simple.   Suppose that for a codeword $g \in \calg$, we see a corrupted version $g + \eta \in \F_2^d$, where $\eta$ is a noise vector.  As described above, $g$ is made up of a prefix from $w_{i+1}$ and a suffix from $w_i$, for some $i$.  Let $h \in [d]$ be the index where $g$ ``crosses over'' from $w_{i+1}$ to $w_i$.  Notice that, as this crossover point can only be in one place, at least one of the two codewords of $\calc$ appearing in $g$ will be complete, and equal to either $c_i$ or $c_{i+1}$.  Thus, if we could identify where the crossover point $h$ was, then we could use $\calc$'s decoder to decode whichever the complete $\calc$-codeword was to identify $i$; and then use our knowledge of where $h$ is to complete the decoding.  The simple observation behind our construction is that, because the strings $s_i$ (which are either all zeros or all ones) flip with the parity of $i$, we \emph{can} tell (approximately) where $h$ was!  Indeed, these strings will be all zeros before $h$ and all ones after $h$, or vice versa.  Of course, some noise will be added, but provided that the length of the strings $s_i$ are long enough, we will still be able to approximately locate $h$ with high probability.

However, there are several challenges to implementing this simple idea.  For example, given $i$ and $h$, how do we efficiently compute $j$? (Here is where the fact that we ordered $\calc$ carefully comes in; it's not trivial because the number of codewords of $\calg$ inserted between $w_i$ and $w_{i+1}$ depends on $i$, so naively adding up the number of codewords of $\calg$ that come before $w_i$ and then adding $h$ would take exponential time.)   Or, what happens when the crossover point $h$ is very close to the end of $g_j + \eta$?  (Here, it might be the case that we mis-identify $i$; but we show that this does not matter, with high probability, because our final estimate will still be close to $j$ with high probability). 
In the rest of the paper, we show how to deal with these and other challenges.

\section{Preliminaries}
We begin by setting notation.  Throughout, we work with linear codes over $\F_2$, so all arithmetic between codewords is modulo 2.
 For $x,y \in \mathbb{F}_2^\ell$, let $\Delta(x,y)$  denote the Hamming distance between $x$ and $y$. We use $\|x\|$ to denote the Hamming weight of a vector $x \in \F_2^\ell$. 
 For a code $\calc \subseteq \F_2^n$, the minimum distance of the code is given by $D(\calc) := \max_{c \neq c' \in C} \Delta(c,c')$. 
 
 For two strings $s_1$ and $s_2$, we use $s_1 \circ s_2$ to denote the concatenation of $s_1$ and $s_2$. For a string $s \in \F_2^\ell$ and for $i \leq \ell$, we use $\text{pref}_i(s) \in \mathbb{F}_2^i$ to denote the prefix of the string $s$ ending at (and including) index $i$. Analogously, we use $\text{suff}_i(s) \in \mathbb{F}_2^{\ell-i}$ be defined as the suffix of $s$ starting at (and including) index $i$. 
 For an integer $\ell$, we use $[\ell]$ to denote the set $\{1, \ldots, \ell\}$.
 
 For $\ell \in \mathbb{Z}$, let $\maj_\ell : \mathbb{F}_2^\ell \rightarrow \F_2$ be majority function on $\ell$ bits. (In the case that $\ell$ is even and a string $y \in \F_2^\ell$ has an equal number of zeros and ones, $\maj_\ell(y)$ is defined to be a randomized function that outputs $0$ or $1$ each with probability $1/2$.)  We use $\ber(p)$ to denote the Bernoulli-$p$ distribution on $\F_2$, so if $X \sim \ber(p)$, then $X$ is $1$ with probability $p$ and $0$ with probability $1-p$.

Next we define Binary Reflected Codes, a classical Gray code (\cite{gray}; see also, e.g.,~\cite{knuth}); we will use these to define our ordering on $\calc$.  
\begin{definition}[Binary Reflected Code, \cite{gray}]
        Let $k$ be a positive integer. The \textbf{Binary Reflected Code (BRC)} is a map $\calr_k: \sett{0,\ldots, 2^k-1} \rightarrow \mathbb{F}_2^k $ defined recursively as follows.
    \begin{enumerate}\label{def:BRC}
    \item For $k = 1$,  $\calr_1(0) = 0 $ and $\calr_1(1) = 1 $.
    \item For $k > 1$, for any $i\in \sett{0,\ldots,2^k-1}$, 
    \begin{itemize}
       \item If $  i < 2^{k-1} $, then $ \calr_k(i) = 0 \circ R_{k-1}(i)$
       \item If $ i\geq 2^{k-1}$, then $\calr_k (i) = 1 \circ\overline{ R}_{k-1} ( 2^{k-1} - (i- 2^{k-1})) = 1 \circ\overline{ R}_{k-1} ( 2^{k} - i).$ 
    \end{itemize} 
\end{enumerate}
\end{definition}
It is not hard to see that for any two successive integers $i$ and $i+1$, the encoded values $\calr_k(i)$ and $\calr_k(i+1)$ differ in exactly one bit.
We will need one more building-block, the Unary code.
\begin{definition}[Unary code] The \textbf{Unary code} $\calu \subseteq \F_2^\ell$ is defined as the image of the encoding map $\enc{\calu}:\{0,\ldots, \ell\} \to \F_2^\ell$ given by
$ \enc{\calu}(v) :=  1^{v}\circ 0^{\ell -v }.$
The decoding map $\dec{\calu}: \F_2^\ell \to \{0,\ldots, \ell\}$ is given by
\[ \dec{\calu} (x)  =  \mathrm{argmin}_{v\in {\sett{0,\ldots,\ell}}} \Delta(x ,\enc{\calu}(v)).\]
\end{definition}

Next, we define the \emph{failure probability} of a code $\calc$.
\begin{definition}\label{def:prob-fail} 
Let $\calc \subseteq \mathbb{F}_2^n$ be a code with message length $k$ and encoding and decoding maps $\dec{\calc}$ and $\enc{\calc}$ respectively.
The probability of failure of $\calc$ is
    \begin{equation*}
        \pfail( \calc) = \max_{v\in \mathbb{F}_2^k } P [ \dec{\calc} ( \enc{\calc} (v) + \eta_p) \neq v ], 
    \end{equation*}
    where the probability is over a noise vector $\eta_p \in \F_2^n$ with $\eta_p \sim \mathrm{Ber}(p)^n$.
\end{definition}

\section{Construction}\label{sec:construction}

We recall the high-level overview of our construction from the introduction: 
To construct $\calg$ we will start with a base code $\calc$ where $ \calc \subseteq \F_2^n$, which we will order in a particular way (Definition~\ref{def:order-c}).  Then we construct an intermediate code $\calw = \{w_0, \ldots, w_{2^k-1}\} \subseteq  \mathbb{F}_2^d$  \ by transforming the codewords of $\calc$ (Definition~\ref{def:calw}); the codewords of $\calw$ inherit an order of $\calc$. 
Finally, we create final code $\calg \subseteq \F_2^d$ by adding new codewords that ``interpolate'' between the codewords of $\calw$ so that it satisfies the Gray code condition (Definition~\ref{def:calg}).
We discuss each of these steps in the subsequent subsections. 

\subsection{Base Code $\calc$}\label{sec:basecode}

Given a base code $\calc \subset \F_2^n$, we define an ordering on the elements of $\calc$ as follows.

\begin{definition}\label{def:order-c}[Ordering on $\calc$]
Let $\calc \subseteq \F_2^n$ be a linear code  with block length $n$ and dimension $k$.  Let $A_\calc \in \mathbb{F}_2^{k \times n}$ be a generator matrix for $\calc$, and let $a_i$ denote the $i$-th row of $A_\calc$. 
Given $i \in \sett{0,\ldots, 2^k-1}$, define $z_i$ to be the unique integer so that $\calr_k(i)[z_i] \neq \calr_k(i+1)[z_i]$.\footnote{As noted after  Definition~\ref{def:BRC}, $\calr_k(i)$ and $\calr_k (i+1)$ differ in only one bit, so $z_i$ is well-defined.} Let $c_0 = 0^n$. Then, for all $i \in \sett{1,\ldots, 2^k-1}$, the $i$-th codeword of $\mathcal{C}$ is defined by
\[
c_i = c_{i-1} + a_{z_i}.
\]
\end{definition}

Our next lemma establishes that indeed this ordering hits all of the codewords.
\begin{lemma}
Let $\calc$ be a linear code, and consider the ordering defined in Definition~\ref{def:order-c}.
 For every $c\in \calc$, there is a unique index $i\in \sett{0,\ldots,2^k-1}$ such that  $c_i = c $. 
\end{lemma}

\begin{proof}
    Observe that, by construction, we have
    \[ c_i = \calr_k(i)^T A_\calc.\]
    Since $\calr_k: \{0,\ldots, 2^k-1\} \to \F_2^k$ is a bijection and $A_\calc$ is full rank, this implies that each codeword in $\calc$ is uniquely represented as some $c_i$ for $i \in \{0, \ldots, 2^k-1\}.$
\end{proof}

\subsection{Intermediate Code $\calw$}\label{sec:W}
Next, we describe how to generate our intermediate code $\calw$.
\begin{definition} 
\label{def:calw}
    Let $\calc \subseteq \mathbb{F}_2^n$ be a linear code of dimension $k$. Let $(c_0, \ldots, c_{2^k-1})$ denote the ordering of codewords in $\calc$ as per Definition~\ref{def:order-c}. Let $d= 2 n + 3 D(\calc)$. The intermediate code $\calw$, along with its ordering, is defined as follows. For each $i \in \{0, \ldots, 2^k-1\}$, define $w_i \in \mathbb{F}_2^{d}$ by the equation
    \begin{equation}
        w_i =\begin{cases}
             0^{D(\calc)} \circ c_i\circ   0^{D(\calc)} \circ c_i \circ 0^{D(\calc)}   & \text{if } i \text{ is even} \\
             1^{D(\calc)} \circ c_i \circ 1^{D(\calc)} \circ c_i  \circ 1^{D(\calc)} & \text{if } i \text{ is odd}
        \end{cases}
    \end{equation}
Then, $\calw$ is the subset of $\mathbb{F}_2^{d}$ defined by $\calw = \{ w_i \,:\, i \in \{1, \ldots, 2^k-1\} \}$, where the code $\calw$ is ordered as $(w_0, \ldots, w_{2^k-1})$. 
\end{definition}

\subsection{Final Code $\calg$}

To create our robust Gray code $\calg$, given any two consecutive codewords in $\calw$, we inject extra codewords between them to create $\calg$, as follows. 

\begin{definition}[Definition of our robust Gray code $\calg$; and the parameters $r_i, h_{i,j}$]
\label{def:calg}
Let $\calw \subseteq \F_2^d$ be a code defined as in Definition \ref{def:calw}. For each $i \in [2^k]$,  define $r_i = \sum_{\ell=1}^{i} \Delta(w_{\ell-1},w_{\ell})$, and let $N = r_{2^k}$. For $i \in [2^k]$ and $1\leq j < \Delta (w_i, w_{i+1})$, let $h_{i,j} \in \{0,\ldots, d-1\}$ be the $j$-th index where codewords $w_i$ and $w_{i+1}$ differ. 

Define the zero'th codeword of $\calg$ as $g_0 = w_0$.  
Fix $j \in \{1, \ldots, N-1\}$.  If  $j = r_i$ for some $i$, we define $g_j \in \F_2^d$ by $g_j = w_i$.
On the other hand, if 
$j \in (r_i, r_{i+1})$ for some $i$, then we define $g_j \in \mathbb{F}_2^{d}$ as
\begin{equation}
g_j = \text{pref}_{h_{i,j-r_i}}(w_{i+1}) \circ \text{suff}_{h_{i,j-r_i}+1}(w_i).
\end{equation}

Finally, define $\calg \subseteq \F_2^d$ by $\calg = \{g_i \,:\, i \in \{0, \ldots, N-1\}\},$ along with the encoding map $\enc{\calg}: \{0,\ldots, N-1\} \to \F_2^d$ given by $\enc{\calg}(i) = g_i$.

We will also define $h_i =(h_{i,1}, h_{i,2}, \ldots , h_{i,\Delta (w_i, w_{i+1})-1} ) \in \{0,\ldots,d-1\}^{\Delta(w_i, w_{i+1})-1}$ to be the vector of all indices in which $w_i$ and $w_{i+1}$ differ, in order, except for the last one.\footnote{The reason we don't include the last one is because once the last differing bit has been flipped, $g_j$ will lie in $[w_{i+1}, w_{i+2})$, not $[w_i, w_{i+1})$.}
\end{definition} 

It will frequently be useful to be able to locate $j \in [N]$ within a block $[r_i, r_{i+1})$. To that end, we introduce the following notation.
\begin{definition}\label{def:bar}
Let $j \in \{0, \ldots, N-1\}$.  Let $i \in \{0, \ldots, 2^{k}-1\}$ be such that $j \in [r_i, r_{i+1})$.  Then we will use the notation $\bar{j}$ to denote $j - r_i$.  That is, $\bar{j}$ is the index of $j$ in the block $[r_i, r_{i+1})$.
\end{definition}
Note that, in this notation, when $j \in [r_i, r_{i+1})$, the last bit that has been flipped to arrive at $g_j$ in the ordering of $\calg$ (that is, the ``crossover point'' alluded to in the introduction)  is $h_{i,\bar{j}}.$
We make a few useful observations about Definition~\ref{def:calg}.  The first two follow immediately from the definition.
\begin{observation}\label{obs:graycode}
$\calg$ is a Gray code.  That is,
    For any $j \in \{0, \ldots, N-1\}$, we have that $ \Delta( g_j , g_{j+1} ) = 1$.  
\end{observation}
\begin{observation}\label{obs:rate}
Suppose that $\calc \subseteq \F_2^n$ has rate $R = \log_2|\calc|/n$ and distance $o(n)$.  Then the code $\calg$ constructed as in Definition~\ref{def:calg} has rate that approaches $R/2$ as $N \to \infty$.
\end{observation}
\begin{observation}\label{obs:unary}
    Let $g_j \in \calg$, and suppose that $j \in (r_i, r_{i+1})$ for some $i \in \{0, \ldots, 2^k -1\}$.  Then
    \[ (g_j + w_i)[h_i] = \enc{\calu}(\bar{j}),\]
    where $\calu \subset \F_2^{\Delta(w_i, w_{i+1})}$ is the unary code of length $\Delta(w_i, w_{i+1})$.  Above, $(g_j + w_i)[h_i]$ denotes the restriction of the vector $g_i + w_i \in \F_2^d$ to the indices that appear in the vector $h_i$.
\end{observation}
\begin{proof}
    By definition, $h_i$ contains the indices on which $w_i$ and $w_{i+1}$ differ, and also by definition, by the time we have reached $g_j$, the first $j - r_i = \bar{j}$ of these indices have been flipped from agreeing with $w_i$ to agreeing with $w_{i+1}$.  Thus, if we add $g_j$ and $w_i$ (mod 2), we will get $1$ on the first $j-r_i$ indices and $0$ on the on the rest.
\end{proof}

Before we move on, we show Definition~\ref{def:calg} actually defines an injective map.
\begin{lemma}\label{lem:injective}
    Let $\calg$ be a code with encoding map $\enc{\calg}$ be as defined in Definition~\ref{def:calg}. Then $\enc{\calg}$ is injective.  
\end{lemma}

\begin{proof} 
   Assume, for the sake of contradiction, that there are two distinct $j,j^\prime \in \sett{0, \ldots , N-1}$ such that $g_j= g_{j^\prime}$.  Without loss of generality assume that $ j^\prime > j$.  There are three scenarios possible. 
    \begin{enumerate}
        \item  \textbf{Case 1:} Both $j$ and $j^\prime$ are in the interval $[r_i, r_{i+1})$.  Then we claim that $g_j[h_{i,\bar{j'}}] \neq g_{j^\prime}[h_{i,\bar{j'}}]$. The reason is that  $g_j[h_{i,\bar{j'}}] = w_i[h_{i,\bar{j'}}]$ and $g_{j^\prime} [h_{i,\bar{j'}}] = w_{i+1} [h_{i,\bar{j'}}]$; but by definition of $h_{i,\bar{j'}}$,  $w_i[h_{i,\bar{j'}}] \neq w_{i+1}[h_{i,\bar{j'}}]$. Thus, $g_j \neq g_{j^\prime}$.
        
        \item \textbf{Case 2: }$j\in [r_{i-1}, r_i)$ and $j^\prime \in [r_{i}, r_{i+1})$.
            Then $g_j$ is an interpolation of $ w_{i-1}$ and  $w_i$, and $g_{j^\prime}$ is an interpolation of $ w_{i} $ and $w_{i+1}$. Notice that $g_j[d-1] = w_{i-1} [d-1]$ and  $g_{j^\prime} [d-1] = w_{i}[d-1]$,
            as the last index of the codewords $g_j$ and $g_{j'}$ has not been flipped yet.  However, as $i$ and $i-1$ have different parity, $w_{i} [d-1] \neq w_{i-1} [d-1]$, which implies $g_j \neq g_{j'}$.
        
        \item \textbf{Case 3:} $j\in [r_{i}, r_{i+1})$ and $j^\prime \in [r_{i^\prime}, r_{i^\prime+1})$  where $ |i-i^\prime| > 1$. In this scenario, $ \sett{c_i, c_{i+1}} \cap \sett{c_{i^\prime} , c_{i^\prime + 1 } } = \emptyset $. 
        Suppose that neither $h_{i,\bar{j}}$ nor  $h_{i^\prime, \bar{j'}}$ are in $[D(\calc), D(\calc) + n)$.
       Then $\tilde{c}_1(g_j) \in \sett{c_i, c_{i+1}}$ and $\tilde{c_1}(g_{j^\prime}) \in \sett{c_{i^\prime}, c_{i^\prime + 1 } } $ leading to $ \tilde{c}_1(g_j) \neq \tilde{c_1}(g_{j^\prime})$ hence $g_j\neq g_{j^\prime}$. The same holds if neither are in $[2D(\calc) + n, 2D(\calc) + 2n)$, repeating the argument with $\tilde{c}_2$.
       The final sub-case is that $h_{i,\bar{j}}  \in [D(\calc), D(\calc) + n)$ and  $h_{i^\prime, \bar{j'}} \in [2D(\calc) + n, 2D(\calc) + 2n)$ or vice versa.  If this occurs (suppose without loss of generality, it is the first one, not the ``vice versa'' case), then according to Lemma~\ref{lem:repetition-differ}, $ s_1(g_j) \neq s_2(g_j)$, however for $g_{j^\prime}$, $s_1(g_{j^\prime} ) =  s_1(w_{i^\prime + 1 })  = s_2( w_{i^\prime + 1 } ) = s_2(g_{j^\prime})  $.  This implies that either $ s_1(g_j)\neq s_1(g_{j^\prime})$ or $s_2(g_j)\neq s_2(g_{j^\prime})$, which implies that $g_j \neq g_{j'}$, as desired.
    \end{enumerate}

\end{proof}

\section{Decoding Algorithm and Analysis}\label{sec:decalg}

In this section, we define our decoding algorithm and analyze it.
We begin with some notation for the different parts of the codewords $g_j \in \calg$.
For a string $x$, we use $ x[i:i']$ to denote the substring $(x_i, x_{i+1}, \ldots, x_{i'-1})$.  With this notation, for any $x \in \F_2^d$, define the following substrings: 
\begin{itemize}
    \item $s_1 (x) = x [ 0:D(\calc)] $
    \item $ \tilde{c}_1 (x) =x [D(\calc): D(\calc) + n ]$
    \item $ s_2 (x) =x [ D(\calc) + n : 2D(\calc) + n ] $
    \item $ \tilde{c}_2(x) = x[2 D(\calc) + n , 2 D(\calc) + 2 n ]  $  \item  $ s_3 (x) = x [2 D(\calc) + 2 n :3 D(\calc) + 2 n] $. 
\end{itemize}
Notice that if $x \in \calg$, then $ \tilde{c}_1$ and $\tilde{c}_2$ are in locations corresponding to the codewords of $\calc$ that appear in codewords of $\calw$, while $s_1, s_2$, and $s_3$ are in locations corresponding to the $0^{D(\calc)}$ and $1^{D(\calc)}$ strings.

Before we formally state the algorithm (Algorithm~\ref{alg:Dec} below), we prove a few lemmas to motivate its structure.

Our first lemma formalizes the intuition in the introduction that at most one of the ``chunks'' in each codeword $g_j$ is broken up by the crossover point $h_{i,\bar{j}}$.
\begin{lemma}\label{lemma:one-broken-chunk}
    Fix $j \in \{0,\ldots,N-1\}$.  Suppose that $i \in \{0, \ldots, 2^k - 1\}$ is such that $j \in [r_i, r_{i+1})$, so 
    $g_j\in \calg $ can be written as $ g_j = s_1\circ \tilde{c}_1 \circ s_2 \circ \tilde{c}_2 \circ s_3$ as above.   
    Then at most one of the substrings in  $ \cals = \sett{s_1, s_2,s_3, \tilde{c}_1, \tilde{c}_2}$ that is not equal to the corresponding substring in  $w_i$ or $w_{i+1}$.
\end{lemma}
\begin{proof}
    First, suppose that  $j = r_i$.  Then in that case $g_j = w_i$ and all of the substrings in $\mathcal{S}$ are equal to their corresponding substring.  Otherwise, $ j \in (r_i, r_{i+1})$. In that case, $\bar{j} \in (0, r_{i+1} - r_i) = (0, \Delta(w_i, w_{i+1}))$. 
    This means that $h_{i,\bar{j}}$ (the ``crossover point'' for $g_j$) is defined, and indexes a position in $g_j$, and in particular in one of the sub-strings in $\mathcal{S}$.  Then other substrings strictly to the left of $h_{i,\bar{j}} $ are equal to their corresponding substring in $w_{i+1}$; and the ones  strictly to the right are equal to the corresponding substring in $w_i$.
\end{proof}
Using the language of Lemma~\ref{lemma:one-broken-chunk}, we say that a substring in $\mathcal{S}$ that is equal to its corresponding substring in $w_i$ is a \textbf{\chunkn}.
Thus,  Lemma~\ref{lemma:one-broken-chunk} implies that there are at least four {\chunkn}s in any $g_j \in \calg$.
Notice that it is possible that a substring $\tilde{c}_\ell$ is in $\calc$ but is not a \chunkn.
We say that \textbf{all {\chunkn}s are decoded correctly} if, for \chunk of $x$, when we run the corresponding decoder, we get the right answer.  That is, if $\tilde{c}_1(x)$ is a \chunkn, then if we were to run $\dec{\calc}$ on $\tilde{c}_1(x)$ we would obtain $\tilde{c}_1(g_j)$, and similarly for $\tilde{c}_2$; and if $s_1(x)$ is a \chunkn, and we were to run $\maj_{D(\calc)}$ on $s_1(x)$, we would obtain $s_1(g_j)$, and similarly for $s_2$ and $s_3$.

Next, we show that if the ``crossover point'' $h_{i,\bar{j}}$ does not point to one of chunks $s_1,s_2$, or $s_3$, then there are at least two of them that differ. 

\begin{lemma}\label{lem:repetition-differ}
    Let $\calg$ be a code defined as in Definition~\ref{def:calg}. Fix any $g_j \in \calg$ and let $r_i$ be such that $j \in [r_i, r_{i+1})$.  Suppose that $h_{i,{\bar{j}}} \in [D(\calc), D(\calc) + n) \cup [2D(\calc) + n, 2D(\calc) + 2n)$;  that is, $h_{i,\bar{j}}$ indexes a position in $\tilde{c}_{i^\prime}(g_j)$ for some $i^\prime \in \sett{1,2}$. Then $ s_{i^\prime}(g_j) \neq s_{i^\prime+1}(g_j)$.
\end{lemma}
\begin{proof} 
Without loss of generality, suppose that $i'=1$.  
By definition, we have
\[ g_j = \mathrm{pref}_{h_{i,\bar{j}}}(w_{i+1}) \circ \mathrm{suff}_{h_{i,\bar{j}} + 1 }(w_i).\]
In particular, since the ``cut-off'' $h_{i, \bar{j}}$ points to a position within $\tilde{c}_1(g_j)$, we have that both $s_1(g_j)$ and $s_2(g_j)$ are {\chunkn}s, and further $s_1(g_j)$ agrees with $w_{i+1}$, while $s_2(g_j)$ agrees with $w_i$.  Since $i$ and $i+1$ have different parities, either $s_1(g_j) = 0^{D(\calc)}$ and $s_2(g_j) = 1^{D(\calc)}$, or the other way around; in either case, they are different.
The same argument holds when $i'=2$.

\end{proof}

Finally, we break things up into three cases, which will be reflected in our algorithm.  In each case, we can use the pattern of the chunks $s_1(x), s_2(x), s_3(x)$ to get an estimate for $c_i$ or $c_{i+1}$, and bound where the crossover point $h_{i,\bar{j}}$ will be.
\begin{lemma}\label{lem:dec-g}
    Let $g_j\in \calg$ and let $i$ be such that $j\in [r_i,r_{i+1})$.  Let $\eta\sim \ber(p)^d$  where $p\in (0,1/2)$. Let $ x = g_j + \eta$ be a received input. Then define $\vhat_{i^\prime} = \dec{\calc}( \tilde{c}_{i^\prime}(x))$ for $i^\prime \in \sett{1,2}$ and $ b_{i^\prime} = \maj_{D(\calc)} ( s_{i^\prime} (x))$ for $i^\prime\in \sett{1,2,3}$. Assume that all {\chunkn}s are decoded correctly by their corresponding decoder.
    Then the following hold.
    \begin{enumerate}
        \item If $(b_1,b_2,b_3) \in \sett{(1,1,0), (0,0,1)}$, then $\enc{\calc}(\vhat_1) = c_{i+1}$ and $h_{i,\bar{j}} \geq n + D(\calc)$.
        \item If $(b_1,b_2,b_3) \in  \sett{(0,1,1), (1,0,0)}$, then $ \enc{\calc}(\vhat_2) = c_i$ and $ h_{i,\bar{j}} \leq n + 2 D(\calc)$.
        \item If $(b_1,b_2,b_3)\in \sett{ (0,0,0) , (1,1,1)} $, then $ \enc{\calc} (\vhat_1)= \enc{\calc}(\vhat_2)\in \sett{c_i, c_{i+1}}$ and $ h_{i,\bar{j}} \in [0,D(\calc))\cup [d- D(\calc) , d )$.
        Moreover, if $h_{i,\bar{j}} \in [0,D(\calc))$, then $\tilde{c}_1(g_j) = \tilde{c}_2(g_j) = c_i$; and otherwise they are equal to $c_{i+1}$.
    \end{enumerate}

\end{lemma}

    \begin{proof} 
    We address each case individually.
    \begin{enumerate}
        \item If $(b_1,b_2,b_3) \in \sett{(1,1,0), (0,0,1)}$ then we claim that $h_{i,\bar{j}} \geq n + D(\calc)$. Assume otherwise. If $h_{i,\bar{j}}\in [0,D(\calc))$, then $g_j[D(\calc):] =  w_i[D(\calc):]$. This means that $s_2(g_j) = s_3(g_j)$, and are {\chunkn}s. Given the assumption that all the {\chunkn}s are decoded correctly, $b_2 =\maj(s_2(x)) = \maj(s_3(x)) = b_3$ but that contradicts our assumption for this case; so we conclude that $h_{i,\bar{j}} \not\in [0,D(\calc))$. Thus, $h_{i,\bar{j}} \in [D(\calc) , n + D(\calc))$.  But then then $s_1(g_j)$ and $s_2(g_j)$ are {\chunkn}s, and  according to Lemma~\ref{lem:repetition-differ}, $s_1(g_j) \neq s_2(g_j)$. Again using the assumption of correct decoding of all {\chunkn}s, this implies that $b_1 \neq b_2$, which again contradicts our assumption for this case. This establishes the claim that $h_{i,\bar{j}} \geq n + D(\calc)$.  
        
        Finally, the fact that $h_{i,\bar{j}} \geq n + D(\calc)$ implies that $\tilde{c}_1 (g_j) = \tilde{c}_1(w_{i+1}) = c_{i+1}$, and $\tilde{c}_1(g_j)$ is a \chunkn. Using the assumption of correct decoding of all {\chunkn}s, we get that $\enc{\calc} (\vhat_1)=c_{i+1}$
        \item If $(b_1,b_2,b_3) \in  \sett{(0,1,1), (1,0,0)}$, then the conclusion follows by an argument nearly identical to case 1. 
        \item If $(b_1,b_2,b_3)\in \sett{ (0,0,0) , (1,1,1)} $, then we claim that $ h_{i,\bar{j}} \in [0,D(\calc))\cup [d- D(\calc) , d )$. Assume otherwise, then $s_1(g_j) = s_1 (w_{i+1})$ and $s_3(g_j) = s_3 (w_i)$ and they are {\chunkn}s. Now as $i$ and $i+1$ do not have the same parity, $s_1(g_j)\neq s_3(g_j)$. As a result, if all {\chunkn}s are decoded correctly, we have that $ b_1 \neq b_3$, which contradicts our assumption in this case. This proves our claim that $h_{i,\bar{j}} \in [0,D(\calc)) \cup [d-D(\calc), d)$. 
        
        If $h_{i,\bar{j}}\in [0,D(\calc))$ then $\tilde{c}_1(g_j) = \tilde{c}_2(g_j) = c_{i}$;  if $h_{i,\bar{j}} \in [ d- D(\calc), d )  $ then $ \tilde{c}_1 (g_j) = \tilde{c}_2 (g_j) = c_{i+1}$; and in either case both are {\chunkn}s. Using the assumption that all {\chunkn}s are decoded correctly, we see that $ \enc{\calc} (\vhat_1)= \enc{\calc}(\vhat_2)\in \sett{c_i, c_{i+1}}$, as desired.
    \end{enumerate}
    \end{proof}

\subsection{{Decoding Algorithm}}

Before we state our main algorithm (Algorithm~\ref{alg:Dec} below), we include a helper algorithm, \compr \ (Algorithm~\ref{alg:compr}).  This algorithm takes an index $i \in \{0, \ldots, 2^k-1\}$ and returns $r_i$.  Note that this is not trivial to do efficiently: If we wanted to compute $r_i$ directly from the definition, that would require computing or storing $\Delta(w_\ell, w_{\ell+1})$ for all $\ell \leq i$ and adding them up, which may take time $\Omega(2^k)$.  Instead, we do something much faster.
\begin{algorithm}[H]
\caption{\compr}
    \begin{algorithmic}
        \State \textbf{Input:} $ i \in \sett{0,..., 2^{k}-1}$
        \State $\hat{r}_i = 0 $
        \For{$z \in \sett{0,\ldots, k-1} $}
            \State $\hat{r}_i = \hat{r}_i +2\lfloor  \frac{ \max(i- 2^z,0) }{2^{z+1} } \rfloor \cdot \|a_z\|  + 3D(\calc)$
            \Comment{$a_z$ is the $z$'th row of the generator matrix of $\calc$.}
        \EndFor
        \State \textbf{Return:} $\hat{r}_i$
    \end{algorithmic}
    \label{alg:compr}
\end{algorithm}
\begin{lemma}\label{lem:computer}
    The Algorithm \compr\  (Algorithm~\ref{alg:compr}) correctly computes $r_i$.
\end{lemma}
\begin{proof}
    Recall that $r_i = \sum_{\ell=0}^{i-1} \Delta(w_\ell, w_{\ell+1})$.  Consider a fixed difference $\Delta(w_\ell, w_{\ell+1})$.  This is precisely \begin{equation}\label{eq:diff}
    \Delta(w_\ell, w_{\ell+1}) = 2\|a_{z_\ell}\| + 3D(\calc),
    \end{equation} where $z_\ell$ is the unique index so that $\calr_k(\ell)[z_\ell] \neq \calr_k(\ell+1)[z_{\ell}]$: indeed, by Definition~\ref{def:order-c}, $\Delta(c_\ell, c_{\ell+1}) = \|a_{z_\ell}\|$, and from that \eqref{eq:diff} follows from the definition of $\calw$ (Definition~\ref{def:calw}).  Thus, in order to compute
    \[ r_i = \sum_{\ell=0}^{i-1} (2\|a_{z_\ell}\| + 3D(\calc)), \] it suffices to count how often each index $z \in \{0,\ldots, k-1\}$ shows up as some $z_\ell$ in that sum.  This is precisely $\left\lfloor \frac{\max(i-2^z, 0)}{2^{z+1}} \right\rfloor,$ by the definition of $\calr_k$.
\end{proof}

Our final algorithm is given in Algorithm~\ref{alg:Dec}. It is organized into the three cases of Lemma~\ref{lem:dec-g}. To help the reader, we have included comments saying what each estimate ``should'' be.  Here, ``should'' is under the assumption that each \chunk is decoded correctly.

\begin{algorithm}
\caption{$\dec{\calg}$: Decoding algorithm for $\calg$}\label{alg:Dec}
\begin{algorithmic}[1] 
    \State \textbf{Input:} $x = g_j + \eta \in \F_2^d$
    \State \textbf{Output:} $\hat{j} \in [N]$

    \For{ $\ell \in \sett{1,2}$}\label{line:start-small-dec}
    \State $\hat{v}_\ell = \dec{\calc} 
    (\tilde{c}_\ell(x))$
    \gComment{Decode $\tilde{c}_1(x)$ and $\tilde{c}_2(x)$ individually to obtain $\hat{v}_1, \hat{v}_2 \in \sett{0,\ldots, 2^k -1}$.}
    \EndFor
    
    \For{ $\ell \in \sett{1,2,3}$ }
    
        \State $b_\ell= \maj_{D(\calc)} ( s_\ell(x))$ \label{line:get_bi}
        \gComment{Decode each $s_\ell(x)$ to obtain $b_\ell \in \sett{0,1}$.}
    \EndFor \label{line:end-small-dec}
    \State \gComment{Below, in the comments we note what each value ``should'' be. 
 This is what these values will be under the assumption that each \chunk is decoded correctly.}

    \If{$(b_1,b_2,b_3) \in \sett{(1,1,0), (0,0,1)}$}
    \gComment{\textbf{Case 1:} $\enc{\calc} (\vhat_1)$ should be $c_{i+1}$}
    \label{line:case1}
        \State $\hat{\iota} = \calr_k^{-1}( \vhat_1)$
        \gComment{ $\hat{\iota}$ should be $i+1$}
        \State $\vhat = \calr_{k}(\hat{\iota} - 1 ) $
        \gComment{$\vhat$ should be the BRC corresponding to $i$}
        \State $\chat_1 = \enc{\calc} (\vhat)$
        \gComment{ $\chat_1$ should be $c_i$}
        \State $\chat_2 = \enc{\calc} (\vhat_1) $
        \gComment{$\chat_2$ should be $c_{i+1}$}
        \State $\what_1 = \enc{\calw} ( \chat_1)$ \label{line:case1-wi}
        \gComment{$\hat{w}_1$ should be $ w_i $}
        \State $\what_2 = \enc{\calw} ( \chat_2)$
        \gComment{$\what_2$ should be $w_{i+1}$} 
        \State $H' = \{\ell \in h_{\hat{\iota} - 1} \,:\, \ell \geq n + D(\calc)\}$
        \gComment{Should be the set of indices that appear in $h_i$ that are at least $n + D(\calc)$.}

        \State $ u = \dec{\calu} ( x [ H^\prime ] + \what_1[H^\prime] )$\label{line:est1}
        \gComment{$u$ is an estimate of $h_{i,\bar{j}} - \Delta(c_i, c_{i+1})-D(\calc)$}\label{line:case1-u}
        \State $ \jhat = u + D(\calc) + \Delta ( \chat_1, \chat_2)  + \compr (\hat{ \iota}-1 )$

    \ElsIf { $(b_1,b_2,b_3) \in  \sett{(0,1,1), (1,0,0)}$}
    \gComment{\textbf{Case 2:} $\enc{\calc} (\vhat_2)$ should be $c_i$}
    \State $\hat{\iota} = \calr_k^{-1} (\vhat_2)$ 
    \gComment{$\hat{\iota}$ should be $i$}
    \State $\vhat = \calr_k (\hat{\iota} + 1 ) $ 
    \gComment{$\vhat$ should be the BRC corresponding to $ i + 1 $}
    \State $\chat_1 = \enc{\calc} (\vhat_2)$ 
    \gComment{$\chat_1$ should be $c_i$}
    \State $\chat_2 = \enc{\calc} (\vhat)$
    \gComment{$\chat_2$ should be $c_{i+1}$}
    \State $\what_1 = \enc{\calw} ( \chat_1)$ \label{line:case2-wi} 
    \gComment{$\what_1$ should be $ w_i $}
    \State $\what_2 = \enc{\calw} ( \chat_2)$
    \gComment{$\what_2$ should be $w_{i+1}$}
     \State $H' = \{\ell \in h_{\hat{\iota} - 1} \,:\, \ell < n + 2D(\calc)\}$
        \State \gComment{$H'$ should be the set of indices that appear in $h_i$ that are less than $n + 2D(\calc)$.}

    \State $ u = \dec{\calu} ( x [ H^\prime ] + \what_1[H^\prime] )$ \label{line:est2}
    \gComment{$u$ is an estimate of $h_{i,\bar{j}} \leq 2D(\calc) + n$} \label{line:case2-u}
    \State $\jhat = u + \compr(\hat{\iota})$ 

    \ElsIf{ $(b_1,b_2,b_3)\in \sett{ (0,0,0) , (1,1,1)}$}
    \State\gComment{\textbf{Case 3:} $\enc{\calc}(\hat{v}_1)$ and $\enc{\calc}(\hat{v}_2)$ should be equal to each other, and to either $c_i$ or $c_{i+1}$, but we need to figure out which one.}
    \State $\hat{\iota} = \calr_k^{-1} (\vhat_1)$
    \gComment{$\hat{\iota}$ should be either $i$ or $i+1$}
    \State $\vhat = \calr_k (\hat{\iota} - 1 ) $
    \gComment{ $\vhat$ should be BRC encoding of $i$ or $i-1$ depending on $\hat{\iota}$}
    \State $\chat_1 = \enc{\calc} (\vhat)$ 
    \gComment{$\chat_1= c_i$ if $\hat{\iota} = i+1$ and $\chat_1 = c_{i-1}$ if $\hat{\iota} = i$  }
    \State $\chat_2 = \enc{\calc} (\vhat_1)$
        \gComment{$\chat_2= c_{i+1}$ if $\hat{\iota} = i+1$ and $\chat_1 = c_{i}$ if $\hat{\iota} = i$  }
    \State $\what = \enc{W} ( \chat_2)$ \label{line:case3-what}
    \gComment{$\what = w_{i+1}$ if $\hat{\iota} = i+1$ and $\what = w_i$ if $\hat{\iota} = i$  }
    \State $u_1 = \dec{\calu} ( x[ < D(\calc)] + b_1^{D(\calc)} ) $  \label{line:est3}
    \gComment{$u_1$ is an estimate of $h_{i,\bar{j}} < D(C)$ assuming  $\what = w_i$}
    \State $u_2 = \dec{\calu} (x[> 2 D(\calc) + 2 n ] + \bar{b}_1^{D(\calc) })$ \label{line:est4}\label{line:case3-u}
\State    \gComment{$u_2$ is an estimate of $h_{i,\bar{j} } - 2D(\calc) - 2\Delta(c_i, c_{i+1})$ assuming $\what = w_{i+1}$}
    \State $\jhat_1 = u_1 + \compr  (\hat{\iota})$ \label{line:case3-j1}
    \gComment{ Estimate for $j$ assuming $\what = w_i$}
    \State $\jhat_2 = u_2 + 2D(\calc) + 2 \Delta (\chat_1, \chat_2 ) + \compr (\hat{\iota}-1) )$ \label{line:case3-j2} 
        \gComment{ Estimate for $j$ assuming $\what = w_{i+1}$}
    \State $\ghat_1 = \enc{\calg} (\jhat_1)$
    \State $\ghat_2 = \enc{\calg} (\jhat_2)$
    \State $\jhat = \min_{j^\prime \in \sett{\jhat_1 ,\jhat_2} } \Delta ( x, \ghat_{j^\prime})$\label{line:case3-jhat}
    
    \EndIf
    
\end{algorithmic}
\end{algorithm}

\subsection{Analysis}\label{sec:analysis}

Next, we analyze the correctness and running time of Algorithm~\ref{alg:Dec}.
We begin with the running time.

\begin{lemma}\label{lem:running time}
Let $\calc \subseteq \F_2^n$ be a constant rate code.
    Suppose that $\dec{\calc}$ runs in time $T_{\dec{\calc}}(n)$, and $\enc{\calc}$ runs in time $T_{\enc{\calc}}(n)$, and that $D(\calc) = o(n)$.
  Let $A_{\calc}$ be the generator matrix of $\calc$, with rows $a_z$ for $z \in \{0,\ldots, 2^k-1\}$.  Suppose that $\|a_z\|$ can be computed in time $O(1)$.  Then the running time is of $\dec{\calg}$ is 
    $$O(T_{\dec{\calc}}(d) + T_{\enc{\calc}}(d) + d),$$
    and
the running time of $\enc{\calg}$ is $O(T_{\enc{\calc}}(d)).$

\end{lemma}
\begin{remark}[Time to compute $\|a_z\|$]
We note that if $\calc$ is, say, a Reed-Muller code $\calrm(r,m)$, then indeed, given $z$, $\|a_z\|$ can be computed in time $O(1)$: if the binary expansion of $z$ has weight $t \leq r$, then the corresponding row has weight $2^{m-t}$.
For codes that may not have closed-form expressions for their generator matrices, we can pre-compute each $\|a_z\|$ (in total time $O(d^2)$) and store them to allow for $O(1)$ lookup time.\footnote{
 If a lookup table is not desirable and the $\|a_z\|$ cannot otherwise be computed on the fly, then our algorithm still works, and $\dec{\calg}$ runs in time at most
    $O(T_{\dec{\calc}}(d) + T_{\enc{\calc}}(d) + d^2),$
    where we recall that $d = \Theta(n)$ and $O(\log N)$. 
}
\end{remark}

\begin{proof}[Proof of Lemma~\ref{lem:running time}]
    As we are assuming that $D(\calc) = o(n)$, finding the $b_\ell$ also takes time $o(n)$ and is negligible.  Among the other steps, the only non-tivial ones are running the encoding and decoding maps for $\calc$ (each of which happens $O(1)$ times); running \compr\  (which takes time $O(k) = O(d)$ if $\|a_z\|$ can be computed in time $O(1)$); and running $\calr_k$ and $\calr_k^{-1}$, which can be done in time $O(k) = O(d)$.
\end{proof}

Next, we move on to the analysis of the correctness and failure probability of Algorithm~\ref{alg:Dec}.  The final statement is Theorem~\ref{thm:main} below, but we will need several lemmas first.  We begin by showing that, if all {\chunkn}s are decoded correctly, then $g_{\hat{j}}$ is equal to $g_j$ on the portion of indices where the crossover point $h_{i,\bar{j}}$ is guaranteed not to be.
\begin{lemma}\label{lem:chunks-show-up}
    Let $x = g_j +\eta $ be the noisy input to $\dec{\calg}$, where $g_j \in \calg$ and $\eta \sim \ber(p)^n$ and let $\jhat = \dec{\calg}(x).$ Assume that all {\chunkn}s are decoded correctly in Lines~\ref{line:start-small-dec} to \ref{line:end-small-dec}.
    Define $\cali\subset \sett{0,\ldots, d-1}$ as the set of indices that $h_{i,\bar{j}}$ can be equal to depending on the pattern of $(b_1,b_2,b_3)$ according to Lemma~\ref{lem:dec-g}.\footnote{That is, if $(b_1, b_2, b_3) = (0,0,1)$ or $(1,1,0)$, then $\cali = [n + D(\calc), 2n + 3D(\calc))$, and so on.} Then $g_{\jhat} [\bar{\cali}] = g_j[\bar{\cali}]$.
\end{lemma}

\begin{proof}
    First notice that the indices in $\bar{\cali}$ are indices corresponding to a subset $S\subseteq \sett{s_1,s_2,s_3, \tilde{c}_1, \tilde{c}_2}$. As $h_{i,j}\neq \bar{\cali}$ then all chunks in $S$ are {\chunkn}s. Given that full chunks are decoded correctly, then we know that for $\tilde{c}_i\in S$, $\enc{\calc}(\vhat_i) = \tilde{c}_i(g_j)$ and for $s_i \in S $ we have that $b_i^{D(\calc)} = s_i(g_j)$. As a result the decoder fixes the values of these indices and only estimates the values of the rest of the bits in lines~\ref{line:est1}, \ref{line:est2}, \ref{line:est3}, and \ref{line:est4}. Thus, $g_j[\bar{\cali}] = g_{\jhat}[\bar{\cali}]$.
\end{proof}
\begin{lemma}\label{lem:alg-case12}
    For a $i\in \sett{0,\ldots, 2^k-2}$ let $j$ be such that $j\in [r_i ,r_{i+1})$ and  $ h_{i,\bar{j}}\in [ D(\calc) ,d-D(C))$. Let $x= g_j + \eta $ be the noisy input for $\dec{\calg}$ and $\jhat$ be the estimate given by $\dec{\calg}$. Assuming that all {\chunkn}s are decoded correctly, then $ \jhat \in [r_i, r_{i+1})$, Moreover, $|j- \jhat| = \Delta (g_j, g_{\jhat})$.
\end{lemma}
\begin{proof}

We first claim that either Case 1 or Case 2 of Lemma~\ref{lem:dec-g} has occurred.  Indeed, 
the fact that $h_{i,\bar{j}} \in [D(\calc), d - D(\calc))$ implies that $s_1(g_j)$ and $s_3(g_j)$ are both {\chunkn}s, and our assumption that each \chunk is correctly decoded implies that $s_1(g_j) = s_1(w_{i+1})$ while $s_3(g_j) = s_3(w_{i})$.  As $i$ and $i+1$ have different parities, $s_1(g_j) \neq s_3(g_j)$, which implies that we are in Case 1 or Case 2 of Lemma~\ref{lem:dec-g}.

Next, we establish that the estimate $\hat{j}$ returned by the algorithm in Cases 1 or 2 satisfies $\hat{j} \in [r_i, r_{i+1})$.  
Suppose without loss of generality that we are in Case 1, so $(b_1, b_2, b_3) = (0,0,1)$ or $(1,1,0)$ (Case 2 is symmetric). 
We first go through Case 1 of Algorithm~\ref{alg:Dec}, which starts at Line~\ref{line:case1}.  Since we are in Case 1 of Lemma~\ref{lem:dec-g}, that lemma implies that $\enc{\calc}(\hat{v}_1) = c_{i+1}$ and that $h_{i,\bar{j}} \geq n + D(\calc)$.
Thus, in the first case in Algorithm~\ref{alg:Dec}, under the assumption that all {\chunkn}s are correctly decoded, we have $\hat{\iota} - 1 = i$, $\hat{c}_1 = c_i$,  $\hat{c}_2 = c_{i+1}$, $\hat{w}_1 = w_i$, and $\hat{w}_2 = w_{i+1}$.
At the end of this case, the final estimate $\hat{j}$ is set to be 
\[ \hat{j} = u + D(\calc) + \Delta(\hat{c}_1, \hat{c}_2) + \compr(\hat{\iota} - 1).\]

By the above, we have $\hat{\iota}-1 = i$, so by Lemma~\ref{lem:computer}, $\compr(\hat{\iota} - 1) = r_i$.  
Note also that 
$\Delta(\calc) = \Delta(s_1(w_i), s_1(w_{i+1}))$.
Plugging in to our expression for $\hat{j}$ and subtracting $r_i$ from both sides, we have
\begin{equation}\label{eq:hatjdiff} \hat{j} - r_i = u + \Delta(s_1(w_i), s_1(w_{i+1})) + \Delta(\tilde{c}_1(w_i), \tilde{c}_1(w_{i+1})).
\end{equation}

Now, recall that in Algorithm~\ref{alg:Dec}, we have $u = \dec{\calu}(x[H']+ \hat{w}_1[H'])$, 
where $H'$ is the set of $\ell$ appearing in $h_i$ so that $\ell \geq n + D(\calc)$.  
It thus follows from the definition that 
\[ u \leq |H'| < \Delta(w_i[n + D(\calc):] , w_{i+1}[n + D(\calc):]), \]
from the definition of $H'$.
Plugging this into \eqref{eq:hatjdiff}, we see that
$\hat{j} - r_i < \Delta(w_i, w_{i+1})$,
which implies that
$\hat{j} \in [r_i, r_{i+1})$,
as desired.

Finally, we argue that $|\hat{j} - j| = \Delta(g_j, g_{\hat{j}})$.  Indeed, we may write
\[ g_j + g_{\hat{j}} = (g_j + w_i) + (g_{\hat{j}} + w_i) = (g_j + w_i)[h_i] + (g_{\hat{j}} + w_i)[h_i].\]
By Observation~\ref{obs:unary}, $(g_j + w_i)[h_i] = \enc{\calu}(j)$; and by that observation along with the fact that $\hat{j} \in [r_i, r_{i+1})$, we also have $(g_{\hat{j}} + w_i)[h_i] = \enc{\calu}(\hat{j}). $
Thus, 
\[ \|g_j + g_{\hat{j}}\| = \|(g_j + w_i)[h_i] + (g_{\hat{j}} + w_i)[h_i]\| = \| \enc{\calu}{j} + \enc{\calu}{\hat{j}}\| = |j - \hat{j}|,\]
which finishes the proof of the lemma.

\end{proof} 
\begin{lemma}
\label{lem:alg-case3-dist}
    For $i\in \sett{0,\ldots, 2^k-2}$, let $j$ be such that $j\in [r_i, r_i+D(\calc)) \cup [r_{i+1} -D(\calc), r_{i+1})$. Further  let $x = g_j + \eta $ be the noisy input and  $\jhat$ be the estimate obtained from $\dec{\calg}$. Assuming that all {\chunkn}s are decoded correctly, then either
    \begin{enumerate}
        \item $\jhat \in [r_i, r_{i+1})$ and $\Delta(g_j, g_{\jhat}) = |j - \jhat|$; or
        \item $\jhat \in [r_{i+1} , r_{i+2}) \cup [r_{i-1}, r_{i})$ and $|j-\jhat| \leq 2D(\calc)$ and  $\Delta(g_j, g_{\jhat}) = |j-\jhat|$.
    \end{enumerate}
\end{lemma}
\begin{proof}
Unlike in Lemma~\ref{lem:alg-case12}, it is now the case that any of the three cases in Lemma~\ref{lem:dec-g} could occur.  We first consider Cases 1 and 2, so $(b_1,b_2,b_3)\notin \sett{(0,0,0), (1,1,1)}$ in line~\ref{line:get_bi}.  The proof is quite similar to that of Lemma~\ref{lem:alg-case12}, so we just sketch it here.  
Suppose that we are in Case 1 of Lemma~\ref{lem:dec-g} (Case 2 is symmetric).  In Case 1, we claim that $\hat{j} \in [r_i, r_{i+1})$.  Indeed, since all {\chunkn}s are decoded correctly, as we argued in the proof of Lemma~\ref{lem:alg-case12}, the values $\hat{\iota}, \hat{v}, \hat{c}_\ell, \hat{w}_\ell$ are computed correctly, meaning that they match the values that the comments in Algorithm~\ref{alg:Dec} say they should be.  In particular, as before we have
\begin{align*}
 \hat{j} &= u + D(\calc) + \Delta(\tilde{c}_1, \tilde{c}_2) + \compr(\hat{\iota} - 1) \\
 &= u + D(\calc) + \Delta(c_i, c_{i+1}) + r_i \\
 & < \Delta(w_i, w_{i+1}) + r_i \\
 &<r_{i+1}.
 \end{align*}
 This establishes that $\hat{j} \in [r_i, r_{i+1})$, which proves the claim; the rest of this case follows identically to the proof of Lemma~\ref{lem:alg-case12}.

Next we consider Case 3, so  $(b_1,b_2,b_3)\in \sett{(0,0,0), (1,1,1)}$. In this case,  Lemma~\ref{lem:dec-g} (and the assumption that all {\chunkn}s are correctly decoded)
says that $\enc{\calc}(\hat{v}_1) = \enc{\calc}(\hat{v}_2) \in \{c_i, c_{i+1}\}$, which
implies that the value of $\what$ computed in line~\ref{line:case3-what} is either $w_i$ or $w_{i+1}$.
Algorithm~\ref{alg:Dec} then computes two estimates $\jhat_1$ and $\jhat_2$, which are meant to be an estimate of $j$ in the two cases the $\what = w_i$ and $\what = w_{i+1}$; eventually it picks whichever $\hat{j}_\ell$ produces a codeword closer to the received word $x$.

There are two main sub-cases.  In the first, the algorithm guesses correctly, meaning that either $\hat{j} = \hat{j}_1$ and $\hat{w} = w_i$; or that $\hat{j} = \hat{j}_2$ and $\what = w_{i+1}$.
In the other case, the algorithm guesses incorrectly, meaning that $\hat{j} = \hat{j}_1$ and $\hat{w} = w_{i+1}$, or $\hat{j} = \hat{j}_2$ and $\what = w_i$.  We consider each of these in turn.

First suppose that the algorithm guesses correctly.  This case is again quite similar to that of Lemma~\ref{lem:alg-case12}, and we sketch the proof.  Suppose that $\hat{j} = \hat{j}_1$ and $\hat{w} = w_i$; the other way of ``guessing correctly'' leads to a similar argument.
Now, we claim that $\hat{j} \in [r_i, r_{i+1})$.  To see this, notice that in this case, we have
\[ \hat{j} = \hat{j}_1 = u_1 + \compr(\hat{\iota})\]
in Line~\ref{line:case3-j1}.  Since $\what = w_i$, given our assumption that all {\chunkn}s are correctly decoded, it is not hard to see that $\hat{\iota} = i$.  Lemma~\ref{lem:computer} then implies that $\compr(\hat{\iota}) = r_i$, so 
\[ \hat{j} = u_1 + r_i \leq D(\calc) + r_i < \Delta(w_i, w_{i+1}) + r_i = w_{i+1}.\]
This shows that $\hat{j} \in [w_i, w_{i+1})$.  Once we have this, the rest of the argument follows as in Lemma~\ref{lem:alg-case12}.

Now we move onto the second sub-case of Case 3, when the algorithm guesses incorrectly.  Unlike the previous cases we have considered, this is different than Lemma~\ref{lem:alg-case12}, because $\hat{j}$ may end up outside of $[r_i, r_{i+1})$.  Without loss of generality, suppose that $\what = w_i$ but that $\hat{j}$ has been set to $\hat{j}_2$. 
(The other case is similar).  
\begin{claim}\label{claim:ABC}
In this sub-case, the following hold.
    \begin{itemize}
        \item[A.] $\hat{j} \in [r_{i-1}, r_i)$
        \item[B.] $j < r_i + D(\calc)$
        \item[C.] $\hat{j} \geq r_{i-1} + 2\Delta(c_i, c_{i-1}) + 2D(\calc) = r - D(\calc).$
    \end{itemize}
\end{claim}
\begin{proof}
    We begin with B.  
    First, since $\hat{w} = w_i$, this implies that $\tilde{c}_1(g_j) = \tilde{c}_2(g_j) = w_i$, so Lemma~\ref{lem:dec-g} (Case 3) implies that $h_{i,\bar{j}} \in [0, D(\calc))$.  Then since $\bar{j} \leq h_{i,\bar{j}} \leq D(\calc)$, we have
    \[ j = \bar{j} + r_i \leq r_i + D(\calc).\]
    This proved B.

    Next we prove C.  This follows from the computation of $\hat{j}_2$ in Algorithm~\ref{alg:Dec}, along with the assumption that all {\chunkn}s are decoded correctly.  Indeed, we have
    \[ \hat{j} = \hat{j}_2 = u_2 + 2D(\calc) + 2\Delta(\hat{c}_1, \hat{c}_2) + \compr(\hat{\iota}-1).\]
    Since $\hat{w} = w_i$, we are in the case where
    $\hat{\iota} = i$, $\hat{c}_1 = c_{i-1}$, $\hat{c}_2 = c_i$, and the above implies that
    \[ \hat{j} = u_2 + 2D(\calc) + 2\Delta(c_{i-1}, c_i) + r_{i-1}.\]
    The fact that $u_2 \geq 0$ proves inequality in part C.  The equality in part C follows since by the definition of $\calw$, we have
    \[ 2\Delta(c_i,c_{i-1}) + 2D(\calc) = \Delta(w_i, w_{i-1}) - D(\calc).\]

    Finally, we move onto A.  The fact that $\hat{j} \geq r_{i-1}$ follows immediately from C.  The fact that $\hat{j} < r_i$ follows from the fact that, by a computation similar to that above, we have
    \[ \hat{j} = \hat{j}_2 = u_2 + r_i - D(\calc),\]
    which is less than $r_i$ as $u_2 < D(\calc)$.
\end{proof}

Given the claim, we can finish the proof of the lemma in this sub-case.  First, we observe that $|j - \jhat| \leq D(\calc)$; indeed this follows directly from B and C in the claim.

Finally, we show that $\Delta(g_j, g_{\hat{j}}) = |j - \hat{j}|$.  To see this, we first write
\[ \|g_j + g_{\hat{j}} \| = \|(g_j + w_i) + (g_{\hat{j}} + w_i)\|.\]
We claim that $g_j$ and $w_i$ differ on only the indices in $[0, D(\calc))$.  This follows from the fact that $h_{i,\bar{j}} \in [0, D(\calc))$, which we saw in the proof of Claim~\ref{claim:ABC} (part B).  Next, we claim that $g_{\hat{j}}$ and $w_i$ differ only on the indices in $[d - D(\calc), d)$.  Indeed, part C of Claim~\ref{claim:ABC} implies that $r_i - \hat{j} \leq D(\calc)$.  Since $\hat{j} \in [r_{i-1}, r_i)$ (part A of Claim~\ref{claim:ABC}), this means that $\hat{j}$ is in the last chunk (the $s_3$ chunk) of $[r_{i-1}, r_i)$, which proves the claim.
Thus, we have that
\[ \|(g_j + w_i) + (g_{\hat{j}} + w_i)\| = \|g_j + w_i\| + \|g_{\hat{j}} + w_i\|, \]
as the two parts differ on disjoint sets.  Moreover, since $s_1(w_i) = \overline{s_1(w_{i+1})}$, we have $\|g_j + w_i\| = (j-r_i)$, since $w_{i+1}$ and $w_i$ differ on \emph{all} of the first $D(\calc)$ bits, so $g_j$ and $w_i$ differ on all of the first $j-r_i$ bits.  Similarly, $\|g_{\hat{j}} + w_i\| = r_i - \hat{j}$.  Putting everything together, we conclude that 
\[ \Delta(g_j, g_{\hat{j}}) = (j-r_i) + (r_i - \hat{j}) = j - \hat{j}.\]
Since $\hat{j} < j$ in this case (as $\hat{j} \in [w_{i-1}, w_i)$, while $j \in [w_i, w_{i+1})$, this proves the last component of the lemma.
 
\end{proof}

The following lemma is included in~\cite{LP24}.  We include a short proof for completeness.
\begin{lemma}[\cite{LP24}]\label{lem:pfail}
    Let $\calc \subseteq \mathbb{F}_2^n$ be a linear code with message length $k$ and minimum distance $D(\calc)$. Further let $ p \in (0,1/2)$ and $ \eta_p \sim \ber(p)^{D(\calc)}$. Then 
    \begin{equation}\label{eq:lhs}
        \Pr\left[ \| \eta_p \| > \frac{D(\calc)}{2}\right] + \frac{1}{2} \Pr\left[\|\eta_p\| = \frac{D(C)}{2}\right] \leq \pfail(\calc).
    \end{equation}
    
\end{lemma}
\begin{proof}
  Let $\dec{MLD}$ be the maximum likelihood decoder for $\calc$.  That is, given $x \in \F_2^n$, $\dec{MLD}(x) = \mathrm{argmin}_{c \in \calc} \Delta(x,c)$.  If there are multiple codewords $c \in C$ that attain the minimum above, then $\dec{MLD}$ chooses randomly among them. Then 
    \[ \pfail(\calc) \geq \max_{v \in \F_2^k} \Pr[ \dec{MLD}(\enc{\calc}(v) + \eta_p) \neq v ] =: p_{MLD}.\]

 Fix a message $v \in \F_2^k$, and let $c_v = \enc{\calc}(v)$.  Let $v' \in \F_2^k$ be such that $\Delta( c_v, c_{v'} ) = D(\calc)$, where $c_{v'} = \enc{\calc}(v').$
Let $H_{v,v'} = \{ i \in [n] : (c_v)_i \neq (c_{v'})_i \}$ be the set on which $c_v$ and $c_{v'}$ disagree.  
Let $\eta'_p \sim \ber(p)^n$ be a noise vector, and define $\eta_p = \eta'_p[{H_{v,v'}}]$, the restriction of $\eta'_p$ to the positions indexed by $H_{v,v'}$.  Observe that $\eta_p \sim \ber(p)^{D(\calc)}$ as in the lemma statement.
Suppose that $\|\eta_p\| > D(\calc)/2$.  Then 
$ v \neq \dec{MLD}(\enc{\calc}(v)).$
On the other hand, if $\|\eta_p\| = D(\calc)/2$, then with probability at least $1/2$, 
$v \neq \dec{MLD}(\enc{\calc}(v)).$

Together, we conclude that 
\[ \pfail(\calc) \geq p_{MLD} \geq \Pr[ \|\eta_p\| > D(\calc)/2 ] + \frac{1}{2} \Pr[ \|\eta_p\| = D(\calc)/2 ].\]

\end{proof}

\begin{lemma}\label{lem:dist-2dc} Let $ \eta \sim \ber(p)^{D(\calc)}$ be a vector in $\F_2^{D(\calc)}$, for $p\in (0,1/2)$. Let $\caly$ be the repetition code of length $D(\calc)$, so that  $\enc{\caly} : \sett{0,1} \rightarrow \sett{0^{D(\calc)} , 1 ^{D(\calc)}} $. Then for any $b \in \{0,1\}$,

$$
\Pr[ \maj_{D(\calc)} (\enc{\caly}( b ) + \eta ) \neq b ] \leq \pfail(\calc),
$$
where we recall that $\maj_{D(\calc)}$ denotes the majority function.
Above, the randomness is over both the choice of $\eta$ and any randomness that $\maj_{D(\calc)}$ uses to break ties. 

\end{lemma}
\begin{proof}\label{lem:deltaone}
    Fix $b \in \{0,1\}$.  Suppose that $\maj_{D(\calc)}(\enc{\caly}(b) + \eta) \neq b$.  Then either $\|\eta\| > D(\calc)/2$, or else $\|\eta\| = D(\calc)/2$ and the random choice of $\maj_{D(\calc)}$ was incorrect, which happens with probability $1/2$.  Thus,
    \[ \Pr[ \maj_{D(\calc)}(\enc{\caly}(b) + \eta) \neq b ] \leq \Pr\left[ \|\eta\| > \frac{D(\calc)}{2} \right] + \frac{1}{2} \Pr\left[ \|\eta\| = \frac{D(\calc)}{2} \right],\]
    which by Lemma~\ref{lem:pfail} is at most $\pfail(\calc)$.
\end{proof}

Before we prove our main theorem establishing correctness of $\dec{\calg}$ with high probability (Theorem~\ref{thm:main}), we need one more concentration bound.  We use the following from \cite{LP24}.
\begin{lemma}[\cite{LP24}] \label{lem:noise-chernoff}
    Let $x_1, x_2 \in \F_2^d$. Let $\eta \sim \ber(p)^d$ for $p\in(0,1/2)$. Then
    \begin{equation}
        \Pr[\Delta (x_2, x_1 + \eta ) \leq \Delta(x_1, x_1 + \eta )] \leq \exp\left(-\frac{(1-2p)^2}{4p+2} \Delta (x_1, x_2)\right) 
    \end{equation}
\end{lemma}

\begin{lemma}\label{lem:returncloser}
    Let $g_j \in \calg$, and let $\hat{j} = \dec{\calg}(g_j + \eta_p)$.
    Suppose that all {\chunkn}s are decoded correctly.  Then 
    \[ \Delta(g_j + \eta_p, g_{\hat{j}}) \leq \Delta(g_j + \eta_p, g_j).\]
\end{lemma}
\begin{proof}
    By the analysis in the proof of Lemmas~\ref{lem:alg-case12} and \ref{lem:alg-case3-dist}, if all {\chunkn}s are decoded correctly, then all the quantities computed in Algorithm~\ref{alg:Dec} \emph{before} $u$ (in Cases 1 and 2), or before $u_1,u_2$ (in Case 3) are what they ``should'' be.  That is, in Cases 1 and 2, all of the quantities computed before  Lines~\ref{line:case1-u} and \ref{line:case2-u}, respectively are what the comments claim they should be.  In Case 3, all of the comments computed before Line~\ref{line:case3-u}) are what the comments claim they should be. 
 Thus, any error in $\hat{j}$ comes from the estimates of $u$ (in Cases 1 or 2) or $u_1$ and $u_2$ (in Case 3).  
 
 We first work out Case 1.
 First, we observe that it suffices to look only on the set $H'$ defined as in Algorithm~\ref{alg:Dec}.  That is, it suffices to show that
 \[ \Delta((g_j + \eta_p)[H'], g_{\hat{j}}[H'] ) \leq \Delta (g_j + \eta_p)[H'], g_j[H']).\]
 Indeed, $g_j$ and $g_{\hat{j}}$ differ only on $H'$.
Next, recall from Observation~\ref{obs:unary} that $(g_j + w_i)[h_i] = \enc{\calu}(\bar{j})$; that is, restricted to the elements in $h_i$, $g_j + w_i$ has $\bar{j}$ ones followed by all zeros.  Since $\bar{j} \geq n + D(\calc)$ in Case 1 (as shown in the proof of Lemma~\ref{lem:alg-case12}), $(g_j + w_i)[H']$ is $\bar{j} - (\Delta(c_i, c_{i+1}) +D(\calc))$ ones followed by zeros.
Thus,
\[x[H'] + \hat{w}_1[H'] = g_j[H'] + \eta[H'] + w_i[H']\]
is a vector of $\bar{j} - (\Delta(c_i, c_{i+1})+D(\calc))$ ones followed by zeros, plus the noise from $\eta[H']$.
Therefore, 
\[ u = \dec{\calu}(x[H'] + \hat{w}_1[H'])\]
is the decoding of $\enc{\calu}(\bar{j} - (\Delta(c_i,c_{i+1}) + D(\calc))) + \eta_p[H']$. 
For notational convenience, in the following we will introduce the notation $X = D(\calc) + \Delta(c_i,c_{i+1})$.  With this notation (and the fact that $\dec{\calu}$ returns the closest codeword in $\calu$), we conclude that
\begin{equation}\label{eq:dist2}
\Delta( \enc{\calu}(u), \enc{\calu}( \bar{j} - X ) + \eta_p[H']) \leq \Delta ( \enc{\calu}(\bar{j} - X), \enc{\calu}( \bar{j} - X) + \eta_p[H']). 
\end{equation}
We claim that in fact the first of the terms in \eqref{eq:dist2} is equal to $\Delta(g_{\hat{j}}[H'], x[H'])$ and the second is equal to $\Delta(g_j[H'], x[H'])$, which will immediately prove the statement in Case 1.
To see the first part of the claim, first 
notice that in Algorithm~\ref{alg:Dec} (using the fact that all the estimates are what they ``should'' be, as above), we have
$\hat{j} = u + X + r_i$, so $u = \bar{\hat{j}} - X.$
Then 
\begin{align*}
\Delta(g_{\hat{j}}[H'], x[H']) &= \| g_{\hat{j}}[H'] + g_j[H'] + \eta_p[H'] \| \\
&= \| \enc{\calu}(\bar{\hat{j}} - X) + \enc{U}(\bar{j} - X) + \eta_p[H'] \| \\
&= \| \enc{\calu}(u) + \enc{\calu}(\bar{j} - X) +  \eta_p[H'] \| \\
&= \Delta(\enc{\calu}(u), \enc{\calu}(\bar{j}-X) + \eta_p[H'])
\end{align*}
Above, we used the fact that on $H'$, $g_{\hat{j}}$ and $g_j$ differ precisely on the indices between $\bar{\hat{j}} - X$ and $\hat{j} - X$.
This proves the first part of the claim.  For the next part, we have that
\begin{align*}
    \Delta(g_j[H'], x[H']) &= \|\eta[H']\|,
\end{align*}
which is the right hand side of \eqref{eq:dist2}.
This finishes proving the claim, and the lemma in this case.

Case 2 is similar to Case 1 and we omit the analysis.
In Case 3,  we have two candidates, $u_1$ and $u_2$.  By an analysis similar to the above, at least one of the following statements hold:
\begin{itemize}
    \item $\hat{w}_1 = w_i$ and
$ u_1 = \dec{\calu}(\enc{\calu}(j) + \eta_p'),$
or;
\item
$\hat{w}_2 = w_{i+1}$ and
$u_2 = \dec{\calu}(\enc{\calu}(j - 2D(\calc) - 2n) + \eta_p'),$
\end{itemize}
where in both cases $\eta_p'$ has i.i.d. $\ber(p)$ coordinates.
Thus, if the first is the case, we have that 
$$\Delta(g_j + \eta_j, g_{\jhat_1}) \leq \Delta(g_j + \eta_j, g_j)$$
and in the second case we have (again with an analysis similar to that above) that
$$
\Delta(g_j + \eta_j, g_{\jhat_2}) \leq \Delta(g_j + \eta_j, g_j).
$$
Thus, by the definition of $\hat{j}$ in this case (Line~\ref{line:case3-jhat}), we have
\[ \Delta(g_j + \eta_p, g_{\hat{j}}) = \min\{ \Delta(g_j + \eta_j, g_{\jhat_1}), \Delta(g_j + \eta_j, g_{\jhat_2})\} \leq \Delta(g_j + \eta_j, g_j),\]
as desired.
\end{proof}
\begin{theorem}\label{thm:main}
Fix $p \in (0,1/2)$.
    Let $\calc \subseteq \F_2^n$ be a linear code. Let $\calg$ be defined as in Definition \ref{def:calg} from $\calc$. Let $\eta_p \in \F_2^d\sim \ber(p)^d$ Then
\begin{equation}
\Pr[|j - \dec{\calg}(\enc{\calg}(j) + \eta_p)| > t] \leq  \gamma \exp(-\alpha t) + 5P_{\text{fail}}(\calc)
\end{equation}
where $P_{\text{fail}}(\calc)$ is the block error probability of $\calc$, and $\alpha$ and $\gamma$ are constants given by $\alpha = -\frac{(1-2p)^2}{4p+2}$ and $ \gamma = \frac{2}{1-\exp(-\alpha)}$.
\end{theorem}
\begin{proof}
    Let $\cals$ be the event that at least one of the {\chunkn}s in $g_j = \enc{\calg}(j)$ is decoded incorrectly. 
    Let $\jhat  = \dec{\calg}(\enc{\calg}(j) + \eta_p) $. Then 

\begin{align*}
    \Pr[|j -\jhat| > t ] &=  \Pr\left[ | j - \jhat | > t\, |\, \bar{\cals}\right] \cdot \Pr[ \bar{\cals} ]+ \Pr\left[|j-\jhat| > t\, |\, \cals \right] \cdot \Pr[\cals]\\
    &\leq \Pr\left[ | j - \jhat | > t \,|\, \bar{\cals}\right] + \Pr[\cals]
\end{align*}
   
    Let $i \in \{0,\ldots, 2^k-1\}$ be such that $j \in [r_i, r_{i+1})$.
     We will bound each of the two terms above individually.

    We begin with $\Pr\left[ | j - \jhat | > t \,|\, \bar{\cals}\right]$.  There are two scenarios, depending on whether $j$ is safely in the middle of the interval $[r_i, r_{i+1})$ (that is, in the middle three chunks), or if it is near the edges (the outermost chunks).  In either case, if all {\chunkn}s are decoded correctly, then $\Delta(\enc{\calg}({j}), \enc{\calg}({\jhat})) = |j - \hat{j}|$.  Indeed, in the first case this follows from
    Lemma~\ref{lem:alg-case12}, while in the second case it follows from Lemma~\ref{lem:alg-case3-dist}.
    Thus, we see that in either case, the probability $\dec{\calg}$ of returning a particular $\hat{j}$ is
    \begin{align*}
          \Pr\left[\jhat \,|\, \bar{\cals}\right] &\leq \Pr\left[\Delta (\enc{\calg} (j) + \eta_p , \enc{\calg}(\hat{j}) )\leq \Delta (\enc{\calg}(j) + \eta_p , \enc{\calg} (j))\right]\\
            &\leq \exp(-\alpha \Delta(\enc{\calg}(j),\enc{\calg}(\hat{j})) )\\
            &= \exp (-\alpha |j - \hat{j} | ),
        \end{align*}
        Above, the first line follows from Lemma~\ref{lem:returncloser}.
    The second line follows from Lemma~\ref{lem:noise-chernoff}, while the last line follows from the fact that $\Delta(\enc{\calg}{j}, \enc{\calg}(\hat{j})) = |j - \hat{j}|$ as noted above.
 Thus, for any integer $z$ that might be equal to $|j - \hat{j}|$, we have
    \[ \Pr[|j - \hat{j}| = z, \bar{\cals} ] \leq 2\exp(-\alpha z), \]
    where the factor of two comes from the fact that $\hat{j}$ might be either $j + z$ or $j-z$.
    Thus,
        \begin{align*}
            P[|j - \jhat | \geq t \,|\, \bar{\cals}]&\leq \sum_{z = t }^{\infty} 2 \exp(-\alpha z )
              =\frac{2\exp(-\alpha t) }{1-\exp(-\alpha)}.
        \end{align*}

 Now we turn our attention to the second term, $\Pr[\cals]$, the probability that at least one of the {\chunkn}s is decoded incorrectly.   
 The {\chunkn}s $\tilde{c}_\ell(g_j)$ for $\ell \in \{1,2\}$ are codewords in $\calc$ and are decoding using $\dec{\calc}$.  Thus, the probability that either of these chunks (assuming they are {\chunkn}s) are decoded incorrectly is at most twice the failure probability is $\pfail(\calc)$. 
 
 On the other hand, the {\chunkn}s $s_\ell$ for $\ell \in \{1,2,3\}$ are repetition codes of length $\Delta(\calc)$.  Then Lemma~\ref{lem:dist-2dc} that the probability that any of these (assuming it is a \chunkn) is at most $\pfail(\calc)$, so the probability that any of these fail is at most $3\pfail(\calc)$.  Altogether, the probability that at least one \chunk is decoded incorrectly is at most
  $  \Pr[\cals]\leq 5\pfail(\calc).$
Summing both terms up we see that
\begin{equation*}
    \Pr\left[|j-\jhat| \geq t \right] \leq \frac{2\exp(-\alpha t) }{1-\exp(-\alpha)} + 5 \pfail(\calc),
\end{equation*}
which completes the proof of the theorem.
\end{proof}

\section*{Acknowledgment}
MW and DF are partially supported by NSF Grants CCF-2231157 and CCF-2133154.  The first author thanks Rasmus Pagh for bringing our attention to this problem.

\bibliographystyle{IEEEtran}
\bibliography{refs}

\begin{thebibliography}{10}
\providecommand{\url}[1]{#1}
\csname url@samestyle\endcsname
\providecommand{\newblock}{\relax}
\providecommand{\bibinfo}[2]{#2}
\providecommand{\BIBentrySTDinterwordspacing}{\spaceskip=0pt\relax}
\providecommand{\BIBentryALTinterwordstretchfactor}{4}
\providecommand{\BIBentryALTinterwordspacing}{\spaceskip=\fontdimen2\font plus
\BIBentryALTinterwordstretchfactor\fontdimen3\font minus
  \fontdimen4\font\relax}
\providecommand{\BIBforeignlanguage}[2]{{%
\expandafter\ifx\csname l@#1\endcsname\relax
\typeout{** WARNING: IEEEtran.bst: No hyphenation pattern has been}%
\typeout{** loaded for the language `#1'. Using the pattern for}%
\typeout{** the default language instead.}%
\else
\language=\csname l@#1\endcsname
\fi
#2}}
\providecommand{\BIBdecl}{\relax}
\BIBdecl

\bibitem{LP24}
D.~R. Lolck and R.~Pagh, ``Shannon meets gray: Noise-robust, low-sensitivity
  codes with applications in differential privacy,'' in \emph{Proceedings of
  the 2024 Annual ACM-SIAM Symposium on Discrete Algorithms (SODA)}.\hskip 1em
  plus 0.5em minus 0.4em\relax SIAM, 2024, pp. 1050--1066.

\bibitem{ALP21}
M.~Aum{\"u}ller, C.~J. Lebeda, and R.~Pagh, ``Differentially private sparse
  vectors with low error, optimal space, and fast access,'' in
  \emph{Proceedings of the 2021 ACM SIGSAC Conference on Computer and
  Communications Security}, 2021, pp. 1223--1236.

\bibitem{ALS22}
J.~Acharya, Y.~Liu, and Z.~Sun, ``Discrete distribution estimation under
  user-level local differential privacy,'' in \emph{International Conference on
  Artificial Intelligence and Statistics}.\hskip 1em plus 0.5em minus
  0.4em\relax PMLR, 2023, pp. 8561--8585.

\bibitem{ACLST21}
J.~Acharya, C.~Canonne, Y.~Liu, Z.~Sun, and H.~Tyagi, ``Distributed estimation
  with multiple samples per user: Sharp rates and phase transition,''
  \emph{Advances in neural information processing systems}, vol.~34, pp.
  18\,920--18\,931, 2021.

\bibitem{XXW20}
L.~Xiao, X.-G. Xia, and Y.-P. Wang, ``Exact and robust reconstructions of
  integer vectors based on multidimensional chinese remainder theorem
  (md-crt),'' \emph{IEEE Transactions on Signal Processing}, vol.~68, pp.
  5349--5364, 2020.

\bibitem{WX10}
W.~Wang and X.-G. Xia, ``A closed-form robust chinese remainder theorem and its
  performance analysis,'' \emph{IEEE Transactions on Signal Processing},
  vol.~58, no.~11, pp. 5655--5666, 2010.

\bibitem{RP23}
G.~Reeves and H.~D. Pfister, ``Reed--muller codes on bms channels achieve
  vanishing bit-error probability for all rates below capacity,'' \emph{IEEE
  Transactions on Information Theory}, 2023.

\bibitem{A08}
E.~Arikan, ``A performance comparison of polar codes and reed-muller codes,''
  \emph{IEEE Communications Letters}, vol.~12, no.~6, pp. 447--449, 2008.

\bibitem{GX14}
V.~Guruswami and P.~Xia, ``Polar codes: Speed of polarization and polynomial
  gap to capacity,'' \emph{IEEE Transactions on Information Theory}, vol.~61,
  no.~1, pp. 3--16, 2014.

\bibitem{MHR16}
M.~Mondelli, S.~H. Hassani, and R.~L. Urbanke, ``Unified scaling of polar
  codes: Error exponent, scaling exponent, moderate deviations, and error
  floors,'' \emph{IEEE Transactions on Information Theory}, vol.~62, no.~12,
  pp. 6698--6712, 2016.

\bibitem{WLVG23}
H.-P. Wang, T.-C. Lin, A.~Vardy, and R.~Gabrys, ``Sub-4.7 scaling exponent of
  polar codes,'' \emph{IEEE Transactions on Information Theory}, 2023.

\bibitem{gray}
F.~Gray, ``Pulse code communication,'' Mar.~17 1953, uS Patent 2,632,058.

\bibitem{knuth}
D.~E. Knuth, \emph{The art of computer programming, volume 4A: combinatorial
  algorithms, part 1}.\hskip 1em plus 0.5em minus 0.4em\relax Pearson Education
  India, 2011.

\end{thebibliography}

\end{document}